\documentclass[letterpaper,11pt,leqno]{article}
\usepackage{paper}
\usepackage{keytheorems}
\usepackage{xcolor}
\usepackage{tikz}
\newkeytheorem{theorem}
\newkeytheorem{theorem*}[name=theorem, numbered=false]
\newkeytheorem{lemma}

\makeatletter
\newcommand{\asslabel}[2]{\def\@currentlabel{#1}\phantomsection\label{#2}}
\makeatother

\hypersetup{pdftitle={On the vanishing content of GARP in high dimensions}}

\makeatletter
\newcommand\thankssymb[1]{\textsuperscript{\@fnsymbol{#1}}}
\makeatother

\begin{document}

\title{The Empirical Content of Revealed Preference in High Dimensions}

\author{Ian Crawford\thankssymb{1}
\thanks{\thankssymb{1}Department of Economics, University of Oxford,
Manor Road, Oxford, OX1 3UQ, United Kingdom; Nuffield College, New Road, Oxford, OX1 1NF,
United Kingdom. This research was substantially completed while visiting the Research School of Economics at the Australian National University. I am very grateful to Simon Grant and his colleagues at the RSE for their hospitality. \href{mailto:ian.crawford@economics.ox.ac.uk}{ian.crawford@economics.ox.ac.uk}}
\and Longye Tian\thankssymb{2}
%
\thanks{\thankssymb{2}Research School of Economics, The Australian National University, ACT 2601, Australia. \href{mailto:longye.tian@anu.edu.au}{longye.tian@anu.edu.au}}}

\date{}


\begin{titlepage}
\maketitle
\begin{abstract}
We examine how the empirical content of revealed preference theory depends on the dimensionality of the choice environment. While higher-dimensional choice problems may appear more demanding, we show that revealed preference restrictions become less informative. Using Selten’s Area measure, we establish that for any fixed number of observations, the empirical content of GARP converges to zero exponentially fast in the number of goods. We provide complementary proofs based on revealed preference graphs and the Afriat inequalities, and show in simulations calibrated to scanner data that the effect is quantitatively large. We also evaluate potential responses in observational and experimental settings and find that, while these can slow the rate, they do not eliminate this loss of empirical content. 
\end{abstract}

\end{titlepage}

\section{Introduction}
Revealed preference theory offers a nonparametric approach to assessing whether observed choices are consistent with utility maximisation. Its appeal lies in the fact that it imposes testable restrictions without requiring functional form assumptions. A key question is how informative these restrictions are in environments with many goods. 

One might expect that as the number of goods increases, choice problems become more complex, requiring consumers to evaluate trade-offs across many dimensions. This suggests that violations of rationality should become more frequent, and that revealed preference conditions should become more demanding. This paper shows that, in fact, the empirical content of revealed preference theory declines sharply with dimensionality.\footnote{Recent work proposes alternative measures of the restrictiveness of economic models based on their ability to approximate data (see, e.g., \cite{FudenbergGaoLiang_2026}). Our approach follows \cite{selten1991properties} in measuring empirical content as the size of the set of behaviour consistent with the model.} For any fixed number of observations, the fraction of behaviour that satisfies the Generalised Axiom of Revealed Preference (GARP) converges to one exponentially fast in the number of goods. In high-dimensional choice environments, almost all behaviour is consistent with revealed preference. The result follows from a structural observation: violations of GARP correspond to directed cycles in the revealed preference relation. We show that, in high dimensions, the configurations required for such cycles have exponentially small measure. Thus, while choice problems become more complex, the revealed preference restrictions themselves become less informative. We provide two complementary proofs of the main result: a graph-theoretic proof based on cycles in revealed preference, and a parallel proof based on the Afriat linear program.

Our analysis complements recent work developing general frameworks for revealed preference (e.g. \cite{NishimuraOkQuah2017}) and empirical measures of violations (e.g. \cite{EcheniqueLeeShum2011} and \cite{DeanMartin2016}). It focuses on how the empirical content of revealed preference conditions varies with dimensionality. This is particularly relevant for modern datasets such as scanner data, and for experimental designs with many goods.

\section{Illustration}
\cite{selten1991properties} defines an ``Area theory'' as one which predicts a subset of the set of all feasible outcomes.\footnote{He contrasts this with ``point'' theories which predict single outcomes and ``distribution'' theories which predict distributions of outcomes.} The restrictiveness of such a theory is captured by its Area. This measures the size of the predicted set relative to the set of all feasible outcomes. An Area theory is more restrictive the smaller the Area. Following \cite{selten1991properties}, \cite{beattyHowDemandingRevealed2011} define the Area of the theory of competitive neoclassical consumer choice as the size of the set of choices which satisfy GARP relative to the size of the set of choices which exhaust the consumer's budget constraints. 

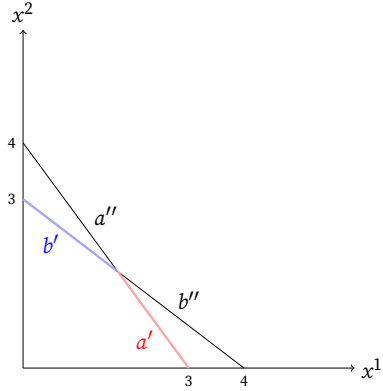
\begin{figure}[!ht]
\begin{center}
\caption{The Area of rational choice in a two-good environment.}
\resizebox{150pt}{150pt}{
\begin{tikzpicture}

  \draw[->] (0,0) -- (6,0) node[right] {$x^1$};
  \draw[->] (0,0) -- (0,6) node[above] {$x^2$};

  \draw (0,3) -- (4,0);
  \draw (0,4) -- (3,0);
  
  \node[below] at (4,0) {\scriptsize{4}};
  \node[below] at (3,0) {\scriptsize{3}};
  \node[left] at (0,4) {\scriptsize{4}};
  \node[left] at (0,3) {\scriptsize{3}};

  \draw [very thick, blue!35] (0,3) -- (12/7,12/7);
  \draw [very thick, red!35] (3,0) -- (12/7,12/7);

   \node[below] at (3,1.5) {{$b^{\prime\prime}$}};
   \node[below] at (0.5,2.5) {\textcolor{blue!95}{$b^{\prime}$}};

   \node[below] at (1.5,3) {{$a^{\prime\prime}$}};
   \node[left] at (2.5,0.5) {\textcolor{red!95}{$a^{\prime}$}};

\end{tikzpicture}
}
\note[Note]{The Area is given by one minus the product of lengths of line segments $a^\prime$ and  $b^\prime$ relative to the lengths of their respective budget constraints.}\label{fig:2D_Selten}

\end{center}
\end{figure} 

To illustrate, consider Figure  \ref{fig:2D_Selten} , which shows two budget constraints in a two-good consumption space. In this case choices on both of the segments $a^\prime$ and $b^\prime$ violate GARP, while all other combinations are consistent with rational choice. The Area of rational choice in this 2-good environment is $A_2  =\frac{40}{49}\approx 0.82$. 

\begin{figure}[!ht]
\begin{center}
\caption{The Area of rational choice in a three-good environment.}
\resizebox{200pt}{200pt}{
\begin{tikzpicture}[x={(-0.5cm,-0.5cm)}, y={(1cm,0cm)}, z={(0cm,1cm)}, >=stealth]

    \draw[->] (0,0,0) -- (6,0,0) node[anchor=north east]{$x^3$};
    \draw[->] (0,0,0) -- (0,5,0) node[anchor=north west]{$x^1$};
    \draw[->] (0,0,0) -- (0,0,5) node[anchor=south]{$x^2$};

    \node[anchor=south west] at (0,0,3){\scriptsize{3}};
    \node[anchor=south west] at (0,4,0){\scriptsize{4}};
    \node[anchor=east] at (5,0,0){\scriptsize{5}};

    \node[anchor=south west] at (0,0,4){\scriptsize{4}};
    \node[anchor=south west] at (0,3,0){\scriptsize{3}};
    \node[anchor=east] at (2,0,0){\scriptsize{2}};
    
    \draw[solid] (0,0,3) -- (0,4,0) -- (5,0,0) -- (0,0,3);    
    \draw[solid] (0,0,4) -- (0,3,0)-- (2,0,0)-- (0,0,4);
    
    \filldraw[
        fill=blue!25,
        draw=blue!25,
        thick,
        opacity=0.8
    ] (0,0,3) -- (0,12/7,12/7) -- (5/7,0,18/7) -- (0,0,3);    
    \filldraw[
        fill=red!25,
        draw=red!25,
        thick,
        opacity=0.8
    ] (0,12/7,12/7) -- (0,3,0) -- (2,0,0)--(5/7,0,18/7) -- (0,12/7,12/7);

    \node[anchor=east] at (2,2,2){\scriptsize{\textcolor{red!95}{$a^{\prime}$}}};
    \node[anchor=west] at (1,1,4){\scriptsize{{$a^{\prime\prime}$}}};

    \node[anchor=east] at (1,1,3){\scriptsize{\textcolor{blue!95}{$b^{\prime}$}}};
    \node[anchor=west] at (4,0,2){\scriptsize{{$b^{\prime\prime}$}}};
\end{tikzpicture}
}
\note[Note]{The Area is given by one minus the product of the areas of the polygon patches $a^\prime$ and  $b^\prime$ relative to the areas of their respective budget constraints.}\label{f:3D_Selten}
\end{center}
\end{figure}
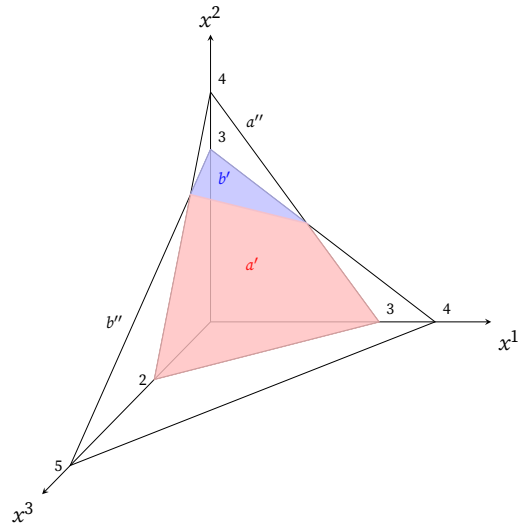 

Now consider adding a third good. Figure \ref{f:3D_Selten} shows the resulting three-dimensional constraints. Choices on both of the patches $a^\prime$ and $b^\prime$ violate GARP. The Area increases to $A_3 \approx 0.95$. Thus, as the dimensionality increases, the empirical content of rational choice is reduced as a larger fraction of possible choices is consistent with GARP.  In this example $A_3>A_2$. In this paper we show that in general $A_K \rightarrow 1$ as $K \rightarrow \infty $ and we establish a bound on the rate. 

\section{Preliminaries}\label{s:prelims}
\subsection{Revealed Preference}
A consumer chooses among $K\ge 2$ goods over $T\ge 2$ observations. At observation $i\in [T]:= \{1,\ldots, T\}$, the consumer faces prices $\bm{p}_i \in \mathbb{R}^K_{++}$ and has income $m_i>0$, and selects a bundle $\bm{x}_i \in \mathbb{R}_+^K$ satisfying $\bm{p}_i\cdot \bm{x}_i = m_i$.

\begin{definition}[Revealed preference and GARP]
Bundle $\bm{x}_i$ is \emph{directly revealed preferred} to $\bm{x}_j$, written $\bm{x}_i R^0 \bm{x}_j$, if $\bm{p}_i \cdot \bm{x}_j \le \bm{p}_i \cdot \bm{x}_i$; strictly so, written $\bm{x}_i P^0 \bm{x}_j$, if the inequality is strict. Let $R$ denote the transitive closure of $R^0$. The dataset $\{\bm{p}_i, \bm{x}_i\}_{i\in [T]}$ satisfies the \emph{Generalised Axiom of Revealed Preference} (GARP) if $\bm{x}_i R \bm{x}_j \implies \text{not } \bm{x}_j P^0 \bm{x}_i$, for all  $i,j\in [T]$.
\end{definition}

GARP, or the feasibility of an equivalent linear programming problem, is a necessary and sufficient condition for a dataset of prices and quantities to be rationalised by a set of preferences represented by a well-behaved utility function.

\begin{theorem*}[Afriat's Theorem]\label{t:Afriats_Theorem}(Afriat (1967), Diewert (1973), Varian (1982)). The following statements are equivalent:
\begin{enumerate}
\item there exists a concave, monotonic, continuous, non-satiated utility function $u(\bm{x})$ which rationalises the data $\{\bm{p}_i,\bm{x}_i\}_{i\in[T]}$
\item there exist real numbers $\{U_i,\lambda_i \}_{i\in[T]}$ which satisfy the inequalities  $U_i \le U_j + \lambda_j \bm{p}_j \cdot (\bm{x}_i - \bm{x}_j)$ and $\lambda_i >0$ for all $i,j\in [T].$
\item the data satisfy GARP.
\end{enumerate}

\end{theorem*}

Write $\bm{r}_i = \bm{p}_i/m_i$ for income-normalised prices, $w_i^k = r_i^k x_i^k$ for the budget share of good $k$ at observation $i$, and $\rho_{ij}^k = r_i^k/r_j^k$ for the ratio of the normalised price of good $k$ between observations $i$ and $j$. Budget shares lie on the $(K-1)$-simplex $\Delta_{K-1} = \{\bm{w}\in \mathbb{R}^K_+: \sum_k w^k =1\}$.  The following observations are useful to note:
\begin{itemize}
    \item The revealed preference relation can be written in terms of normalised prices and demands $\bm{x}_i R^0 \bm{x}_j \iff \bm{r}_{i}\cdot \bm{x}_j \le 1$, or in terms of relative normalised prices and budget shares $ \bm{x}_i R^0 \bm{x}_j \iff \pmb{\rho}_{ij}\cdot \bm{w}_j \le 1 $,  with the strict relation $P^0$ being given by strict inequalities in each case.
    \item The data $\{\bm{p}_i,\bm{x}_i\}_{i\in[T]}$ satisfy GARP iff the data $\{\bm{r}_i,\bm{x}_i\}_{i\in[T]}$ and $\{\pmb{\rho}_{ij},\bm{w}_i\}_{i,j\in[T],i\ne j}$ satisfy GARP.
    \item The linear constraints $U_i \le U_j + \lambda_j \bm{p}_j \cdot (\bm{x}_i - \bm{x}_j)$ and $\lambda_i>0$ for all  $i,j\in [T]$ are feasible iff the linear constraints $ U_i \le U_j + \lambda_j (\bm{r}_j \cdot \bm{x}_i - 1);\lambda_i>0$ , and   $U_i \le U_j + \lambda_j (\pmb{\rho}_{ji} \cdot \bm{w}_i - 1); \lambda_i>0 $ for all $i,j\in [T]$ are feasible.
\end{itemize}

\subsection{Selten's Area measure}

\begin{definition}[Selten Area]
Fix normalised prices. Let $\bm{w}_1,\dots,\bm{w}_T \in \Delta_{K-1}$ be budget-share vectors, where $\bm{w}_i = (w_i^1,\dots,w_i^K)$. Let $\Delta^T_{K-1}$ denote the product simplex 
\[
\Delta_{K-1}^T
=
\underbrace{\Delta_{K-1}\times \cdots \times \Delta_{K-1}}_{T\ \text{times}}.
\]
Let $S_K \subseteq \Delta^T_{K-1}$ denote the set of $(\bm{w}_1,\dots,\bm{w}_T)$ for which GARP is satisfied. The Selten Area is the fraction of the product simplex for which GARP is satisfied:
\[
A_K := \frac{\mathrm{vol}(S_K)}{\mathrm{vol}(\Delta^T_{K-1})}.
\]
\end{definition}

Selten's Area is a deterministic, geometric object. However,  it is hard to calculate analytically in anything other than low-$K$/low-$T$ environments. This is due to the combinatorics involved in enumerating the various permutations and combinations of the many patches on budget constraints which can combine to induce preference cycles and hence violations. Nonetheless, it can be measured numerically using Monte Carlo methods. This approach is standard in geometric probability: the relative volume of a subset of a space can be interpreted as the probability that a uniformly random point lies in that subset. We use a probabilistic representation as a convenient way to characterise the geometry; no stochastic assumptions are required for the main result.

\begin{definition}[Probabilistic representation]\label{d:prob_representation}
Let the budget shares $\bm{w}_1,\dots,\bm{w}_T$ be i.i.d. draws from the uniform distribution on $\Delta_{K-1}$ (with respect to Lebesgue measure on the simplex). Then the Area with $K$ goods can be written as
\[
A_K = \mathbb{P}\big( (\bm{w}_1,\dots,\bm{w}_T) \in S_K \big),
\]
i.e. the probability that $T$ uniformly random budget-share vectors satisfy GARP.
\end{definition}

When calculating the Area in this manner the distinction between weakly preferred and strictly preferred is immaterial in the following sense. Consider a dataset of normalised prices and budget shares distributed uniformly on the simplex $\{ \bm{r}_i,\bm{w}_i \}_{i \in [T]} $ and define the random directed graph $G_K$ on $[T]$ by placing edge $i\to j$ whenever $\pmb{\rho}_{ij}\cdot \bm{w}_j \le 1$ (equivalently, whenever $\bm{x}_i R^0 \bm{x}_j$).   Since $\pmb{\rho}_{ij}$ is not a scalar multiple of $\bm{1}$, the set $\{\bm{w}\in\Delta_{K-1} : \pmb{\rho}_{ij}\cdot\bm{w} = 1\}$ is the intersection of the simplex with a transverse hyperplane and is therefore $(K-2)$-dimensional, hence has Lebesgue measure zero. As the uniform distribution is absolutely continuous with respect to Lebesgue measure, the set has probability zero, so the distinction between weak and strict revealed preference vanishes almost surely. GARP is therefore violated almost surely if $G_K$ contains a directed cycle. 

A directed cycle on $[T]$ is a sequence of distinct vertices $(i_1, i_2, \ldots, i_L, i_1)$ with $L \ge 2$; edge $\ell$ denotes the ordered pair $(i_\ell, i_{\ell+1})$, with $i_{L+1} := i_1$. For each edge $\ell$ and good $k$, write $\rho_\ell^k := \rho_{i_\ell\, i_{\ell+1}}^k$, and define the Carli-type price index\footnote{A true Carli index is the arithmetic mean of un-normalised price ratios whereas in this paper we use income-normalised prices.} along edge $\ell$ as $\bar{\rho}_\ell := \tfrac{1}{K}\sum_{k=1}^K \rho_\ell^k$. The following result shows that variation in the price ratios $\rho_\ell^k$ across goods forces at least one Carli index in any cycle to exceed~1.

\begin{lemma}[manual-num=1, label={l:cycle_mean_inequality}][Carli mean inequality]
$\prod_{\ell=1}^L \bar{\rho}_\ell \ge 1$, with equality if and only if $\rho_\ell^k$ is constant in $k$ for each $\ell$. In particular, if $\rho_\ell^k$ is non-constant in $k$ for some $\ell$, then $\prod_{\ell=1}^L \bar{\rho}_\ell > 1$, and hence $\max_{\ell\in[L]}\bar{\rho}_\ell>1$.
\end{lemma}

\begin{proof}
The proof of this, and all following, results is given in Appendix~\ref{s:proofs}.
\end{proof}

The lemma guarantees that whenever the normalised-price ratios $\rho_\ell^k$ vary across goods on at least one edge of a cycle, at least one Carli index must exceed~1. Since \(\bar{\rho}_\ell\) is the unweighted average of the good-level normalised-price ratios \(\rho_\ell^k = r_{i_\ell}^k/r_{i_{\ell+1}}^k\), the condition \(\bar{\rho}_\ell>1\) means that, averaged across goods, normalised prices are higher at observation \(i_\ell\) than at observation \(i_{\ell+1}\). This will be important in what follows as we will show the probability of an edge forming where the Carli index exceeds one is related to the number of goods. 

\section{Revealed Preference in High Dimensions}\label{s:main}
Since the Area coincides with the probability that uniformly random budget shares satisfy GARP or (equivalently) satisfy the Afriat inequalities, our approach is to ask what happens to this probability as $K \rightarrow \infty$. We first consider GARP, the graph-theoretic condition, and then the linear programming condition. 

For each $K$, we let $\{\rho^k_{ij}\}_{k=1}^K$ denote the good-level price ratios and let $\bm{w}_1, \ldots, \bm{w}_T \in \Delta_{K-1}$ be the corresponding budget shares. We consider a sequence of such environments indexed by $K$, and all assumptions below are imposed uniformly in $K$.   

\subsection{Graph-theoretic approach}

\begin{theorem}[manual-num=1, label={t:garp_high_dim}][Vanishing content of GARP]
Let $T\ge 2$ be fixed. Suppose:\\
\asslabel{A1}{ass:A1}{(\ref{ass:A1}) Boundedness.} There exist constants $0<a\le b<\infty$, independent of $K$, such that $a\le \rho_{ij}^k\le b$ for all $K$, $k\in [K]$, and $i\neq j$.\\
\asslabel{A2}{ass:A2}{(\ref{ass:A2}) Price dispersion.} Price dispersion across goods does not vanish as $K$ grows. Formally: For every directed cycle $(i_1, i_2, \ldots, i_L, i_1)$ on distinct vertices in $[T]$ with $L \ge 2$, there exists $\varepsilon > 0$, independent of $K$, such that for all sufficiently large $K$,
\begin{equation}
    \max_{\ell\in [L]} \bar{\rho}_\ell\ge 1+\varepsilon,
\end{equation}
where $\bar{\rho}_\ell:= \frac{1}{K} \sum_{k=1}^K \rho_{i_\ell i_{\ell+1}}^k$ and $i_{L+1}:=i_1$.

\noindent There exists $c_1>0$, depending only on $a, b, \varepsilon$, such that
\begin{equation}
    A_K \ge 1 - C_T \exp(-c_1K),
\end{equation}
where $C_T = \sum_{L=2}^T \binom{T}{L}(L-1)!$ counts the directed cycles on $[T]$. 
In particular, $A_K \to 1$ as $K\to \infty$.
\end{theorem}

\begin{proof} See Appendix~\ref{a:main_proof}.
\end{proof}

The key geometric object in Theorem 1 is the product of budget-share simplices: $\Delta_{K-1}^T$. A collection of $T$ budget-share vectors corresponds to a point within this product simplex:
\[
(\bm{w}_1,\ldots,\bm{w}_T)\in \Delta^T_{K-1}.
\]
This product simplex is the full feasible choice space: it contains every possible collection of budget-share choices on the \(T\) budget constraints. Inside this space, some points satisfy GARP and some points violate it. Theorem 1 is a statement about the relative volume of these two regions. It says that, holding \(T\) fixed and increasing the number of goods \(K\), the region of $\Delta^T_{K-1}$ that violates GARP becomes exponentially small. Equivalently, the GARP-satisfying region comes to occupy almost all of the feasible choice space.

To see why, consider how a GARP violation appears geometrically. A revealed preference edge \(i\to j\) occurs when $\pmb{\rho}_{ij}\cdot \bm{w}_j\leq 1$. The inequality \(\pmb{\rho}_{ij}\cdot \bm{w}_j\leq 1\) divides the budget share simplex for \(\bm{w}_j\) into two pieces: the part where the edge \(i\to j\) exists (\(\pmb{\rho}_{ij}\cdot \bm{w}_j\leq 1\)), and the part where it does not (\(\pmb{\rho}_{ij}\cdot \bm{w}_j> 1\)).  A GARP violation requires a number of such edges to come together to form a directed cycle, for example
\[
i_1\to i_2\to \cdots \to i_L\to i_1.
\]
Each possible cycle therefore corresponds to a geometric patch inside the product simplex $\Delta^T_{K-1}$: the set of budget-share vectors for which all the inequalities needed for that particular cycle hold simultaneously. The entire GARP-violating region is the union of all these cycle patches.

The assumptions in Theorem 1 ensure that this geometric object is well behaved as \(K\) grows. Assumption A1 rules out relative price ratios that become arbitrarily large or arbitrarily small as goods are added; geometrically, it prevents the slicing hyperplanes \(\pmb{\rho}_{ij}\cdot \bm{w}_j=1\) from becoming degenerate because of extreme price coordinates. Assumption A2 ensures that relative price dispersion does not vanish in high dimensions; in particular, along every possible cycle there remains at least one edge whose average relative normalised price is bounded away from one. Without A2, the cycle inequalities could become nearly neutral in the large-\(K\) limit, and the argument that cycle patches are geometrically small would lose force.

The crucial point behind Theorem 1 is that every cycle patch must pass through at least one geometrically restrictive edge. By the Carli mean argument, together with A2, along any cycle at least one edge has unweighted average relative normalised price bounded above one, 
\[
\bar{\rho}_{i_\ell i_{\ell+1}}
=
\frac{1}{K}\sum_{k=1}^K \rho^k_{i_\ell i_{\ell+1}}
>1+\varepsilon
\]
Call this an ``expensive'' edge. On this expensive edge it must still be the case that, since this is an edge, the weighted average relative prices are less than one: $
\pmb{\rho}_{ij}\cdot \bm{w}_j\leq 1$. This requires that the period-\(j\) budget-share vector must be skewed enough towards goods with low values of $\rho_{ij}^{k}$ in order to pull the weighted average down, away from an unweighted average over \(1+\varepsilon\) , to below \(1\). Such patterned deviation towards some goods and not others is a geometrically selective region relative to the whole simplex of feasible budget. The reason is that, in high dimensions, most of the simplex’s volume lies near allocations in which expenditure is spread thinly across many goods, rather than concentrated on a subset of them. A useful analogy is coin flipping. With many coin flips, outcomes with about half heads dominate not because any one such outcome is intrinsically special, but because there are vastly more approximately balanced outcomes than extreme ones such as all heads or almost all heads. The same logic applies to the simplex. There are relatively few ways, in volume terms, to put a large share of the unit mass on a small set of coordinates; those allocations by definition live near corners, edges, or low-dimensional faces. There are vastly more ways to spread the unit mass thinly across many coordinates. Thus, although the simplex has more and more corners as \(K\) rises, almost none of its volume is close to any particular corner or strongly skewed face. Most of the volume is associated with even allocations in which no small group of goods receives a very large share.

This is why the cycle patches shrink. Any cycle patch must include at least one expensive edge, along which the relevant budget-share vector are sufficiently skewed towards the low-\(\rho_{ij}^k\) goods. That alignment is geometrically restrictive. As \(K\) grows, A1 keeps the individual price-ratio coordinates controlled, while A2 keeps the relevant average price-ratio gap away from zero. Together these conditions allow the concentration-of-volume argument to bite uniformly in \(K\): the relative volume of the part of the simplex with enough alignment to satisfy the expensive edge shrinks exponentially. Since \(T\) is fixed, there are only finitely many possible directed cycles, so the number of cycle patches does not grow with \(K\). Theorem 1 then says that even after taking the union over all possible cycle patches, their total relative volume is at most exponentially small in \(K\). Therefore almost every point in $\Delta^{T}_{K-1}$, in the Selten Area sense, lies outside the cycle patches and hence satisfies GARP.

The proof of Theorem 1 itself uses the language of geometric probability to translate these geometric ideas into probabilistic concepts which are amenable to standard arguments. Under the uniform distribution on the simplex, probability is just a way of measuring relative volume: the probability that a randomly drawn point lies in a given region is exactly the fraction of the total feasible volume occupied by that region. So when the proof writes \(A_K = P((\bm{w}_1,\ldots,\bm{w}_T)\in S_K)\), it is saying that Selten’s Area is the relative volume of the GARP-satisfying region inside the product simplex. Likewise, the event that a particular directed cycle appears in the random graph \(G_K\) is simply the event that the sampled point \((\bm{w}_1,\ldots,\bm{w}_T)\) falls inside the corresponding cycle patch, that is, the geometric region where all of the inequalities defining that cycle hold at once. The Chernoff bound (a standard concentration inequality that shows the probability of a sum or weighted average of independent random terms falling a fixed amount below its mean decays exponentially in the number of terms) then enters as the formal step that turns geometry into a rate: it shows that, for an edge whose Carli mean exceeds one, the set of budget-share vectors on which \(\pmb{\rho}\cdot \bm{w} \le 1\) still holds is not a typical part of the simplex but a thin large-deviation region. In that way, the Chernoff bound turns the geometric fact that the relevant cycle patch lies far from the centre of the simplex into an exponential bound on the probability that the cycle occurs, which is the same as an exponential bound on its relative volume.

It is worth noting the relation between WARP and GARP. In two-good environments, rationalisability is governed by WARP; only two 2-cycles can matter---cycles of length greater than two do not contribute to the Area of WARP. With three or more goods, longer cycles become relevant, and GARP is the correct necessary-and-sufficient condition because it rules out cycles of all lengths. Theorem 1 is a statement about cycles of all lengths: its proof bounds the volume of the regions corresponding to 2-cycles, 3-cycles, 4-cycles, and so on, showing that all of these cycle regions become exponentially small relative to the full product simplex as the number of goods grows. The same concentration mechanism appears at the level of individual edge inequalities: for an edge whose Carli mean exceeds one, the chance that $\pmb{\rho}_{ij}\cdot\bm{w}_j\le1$ holds is exponentially small in $K$. Longer cycles require several such inequalities to hold simultaneously and are therefore no easier to generate. As a result, violations of both WARP and GARP become negligible as  $K$ grows, because concentration effects make all forms of inconsistency rare in high dimensions. However, it is important to note that the result is asymptotic rather than a monotonicity theorem for small steps in $K$. It does not logically force $A_3>A_2$ in every possible environment; rather it points to the mechanisms at play and direction of pressure: adding goods makes the relevant cycle regions smaller relative to the full feasible space. In Section 5 we show evidence from simulations that this effect is already substantial at moderate $K$.

\subsection{The Afriat Linear Program}
We next provide the complementary proof based on the Afriat linear program. Theorem 2 gives the same asymptotic conclusion in the Afriat formulation: instead of cycle patches in the product simplex, it studies feasibility in LP coefficient space. The key point is that the limiting coefficient system is feasible with slack, so small perturbations of the coefficients can be absorbed. The Chernoff bound then shows that, as the number of goods grows, the LP coefficients stay close enough to their means that this slack survives. As a result, infeasibility occurs only on an exponentially small region of coefficient space.

Define the scalar coefficient $e_{ij} := \bm{r}_i \cdot \bm{x}_j = \pmb{\rho}_{ij} \cdot \bm{w}_j$ (so $e_{ii} = 1$); using the normalised-price form of the Afriat inequalities from \S3 (Preliminaries), the LP is a feasibility problem in $2T$ variables with $T(T-1) + T$ constraints, depending on $K$ only through $\{e_{ij}\}_{i\neq j}$. Under uniform budget shares, $\mathbb{E}[e_{ij}] = \bar{\rho}_{ij}$, the Carli price index.

\begin{theorem}[manual-num=2, label={t:lp_vanishing}][Vanishing probability of LP infeasibility]
Let $T\ge 2$ be fixed. Suppose:\\
(\ref{ass:A1}) Boundedness. There exist constants $0<a\le b<\infty$, independent of $K$, such that $a\le \rho_{ij}^k\le b$ for all $K$, $k\in [K]$, and $i\neq j$.\\
\asslabel{A2$'$}{ass:A2p}{(\ref{ass:A2p}) Positive constraint-chain sums.} There exists $\eta>0$, independent of $K$, such that for all sufficiently large $K$ and every sequence of distinct indices $(i_\ell)_{\ell\in [L]}$ with $L\ge 2$ and $i_{L+1}:=i_1$,
\begin{equation}\label{e:chain_assumption}
    \sum_{\ell=1}^L \bigl(\bar{\rho}_{i_\ell i_{\ell+1}} - 1\bigr) \ge \eta.
\end{equation}

\noindent Then there exists $c_2>0$, depending only on $a,b,\eta,T$, such that
\begin{equation}
    A_K = \mathbb{P}\!\left(\text{Afriat Linear Program is feasible}\right) \ge 1 - T(T-1)\exp(-c_2K).
\end{equation}
In particular, the probability of feasibility tends to~1 as $K\to\infty$.
\end{theorem}

\begin{proof} See Appendix~\ref{a:lp_proof}.
\end{proof}

Theorem 2 says that the mean Afriat system is feasible with slack, and random fluctuations around the mean are too small to destroy that slack except on a set of exponentially small volume. The key observation is that under uniform budget shares, each LP coefficient $e_{ij} = \pmb{\rho}_{ij} \cdot \bm{w}_j$ has expected value equal to the Carli index $\bar{\rho}_{ij}$. So in expectation, the random LP is the deterministic LP with $e_{ij}$ replaced by $\bar{\rho}_{ij}$.

The deterministic LP is thus feasible with slack. Its constraints take the form $U_j - U_i \le \bar{\rho}_{ij} - 1$, and A2$'$ guarantees $\sum_\ell (\bar{\rho}_\ell - 1) \ge \eta$ around any cycle. This positive cycle sum is the classical feasibility test for difference constraints, and leaves enough room to tighten every constraint by $\eta_0 := \eta/(2T)$ and still find a solution (Lemma~5).

For the random LP, concentration closes the remaining gap. Each $e_{ij}$ is a weighted average of $K$ price ratios, so a Chernoff bound (Lemma~4) shows it stays within $\eta_0$ of its mean with probability at least $1 - \exp(-c_2 K)$. By a union bound over the $T(T-1)$ pairs, all coefficients stay close enough simultaneously with probability at least $1 - T(T-1)\exp(-c_2 K)$. On this event the slack of the deterministic solution exactly covers the random fluctuation, so the same potentials remain feasible for the random LP and the data are GARP-consistent.

Comparing Theorems 1 and 2 we note the following. First, A2 and A2$'$ are equivalent under A1 and the Carli cycle inequality: A2$'$ implies A2 with $\varepsilon = \eta/T$ (the largest of $L \le T$ summed terms is at least $\eta/T$), and A2 implies A2$'$ via $\prod_\ell \bar{\rho}_\ell \ge 1$ (a symmetric minimisation forces $\sum_\ell (\bar{\rho}_\ell - 1) \gtrsim \varepsilon^2$). So Theorem 2 is a complementary proof rather than a result under a weaker assumption. Second, the prefactor $T(T-1)$ in Theorem 2 grows quadratically in $T$, in contrast to Theorem 1's $C_T = \sum_{L=2}^T \binom{T}{L}(L-1)!$ which grows factorially ($\sim 10^6$ at $T = 10$). For moderate $T$ and finite $K$ the LP bound is therefore tighter in the prefactor; the graph-theoretic bound's advantage emerges asymptotically through a sharper exponent.

We now assess the quantitative importance of these results in empirically relevant environments.

\section{Simulation Studies}

This section has three aims: (i) to quantify the rate at which empirical content declines with dimensionality; (ii) to assess the magnitude of the effect given empirically relevant price distributions; and (iii) to evaluate potential approaches to restoring empirical content in observational and experimental settings.

To provide an empirically-relevant benchmark for simulations we use the distribution of normalised prices in real-world shopping data. The data are the ACNielsen Homescan Panel used in \cite{Aguiar_Hurst_2007}. This is a household-level scanner dataset that records detailed information on grocery purchases. It includes exact prices paid for items at the UPC (barcode) level, along with expenditure and quantity information. The data track shopping behaviour across \textit{all} of the retail outlets participating households use during their time in the panel. The sample covers 2,100 households and nearly one million transactions in the Denver area from 1993 to 1995. These data provide a highly granular, household-level view of prices and spending, making them well-suited for studying price dispersion. We define the budgeting period as a week and normalise the price data by weekly total expenditure. 

The focus is the distribution of the inverse normalised price for each good as this determines the distribution of intercepts of the budget constraints and hence relative price variation and the extent to which constraints cross. We find that in the ACNielsen data this distribution is well approximated by a log-Normal distribution with  $\hat{\mu}=5.69$ and $\hat{\sigma}=1.19$ (see Appendix B for details). Since our simulations draw inverse normalised prices for all goods from a common distribution, the Area is not affected by the location parameter.  In what follows we therefore round the standard deviation and use a $\text{LogN}(0,1)$ as our benchmark. This calibration ensures that the simulated environments match the degree of price dispersion observed in real-world consumption data. We also note that the $\text{LogN}(0,1)$ calibration satisfies A1 (with explicit $a, b$ after a mild tail truncation) and A2 (numerically verifiable along sampled cycles), so the simulations operate within the regime covered by Theorem 1.

\subsection{The effects of varying K and T}

Figure~\ref{fig:AK} illustrates the way in which the Area changes with the number of goods by displaying $A(K)$ curves. 

\begin{figure}[!htb]
    \centering
    \caption{ The relationship between Area and the number of goods.}
    \begin{subfigure}[b]{0.5\textwidth}
        \centering
       \includegraphics[width=1.0\textwidth]{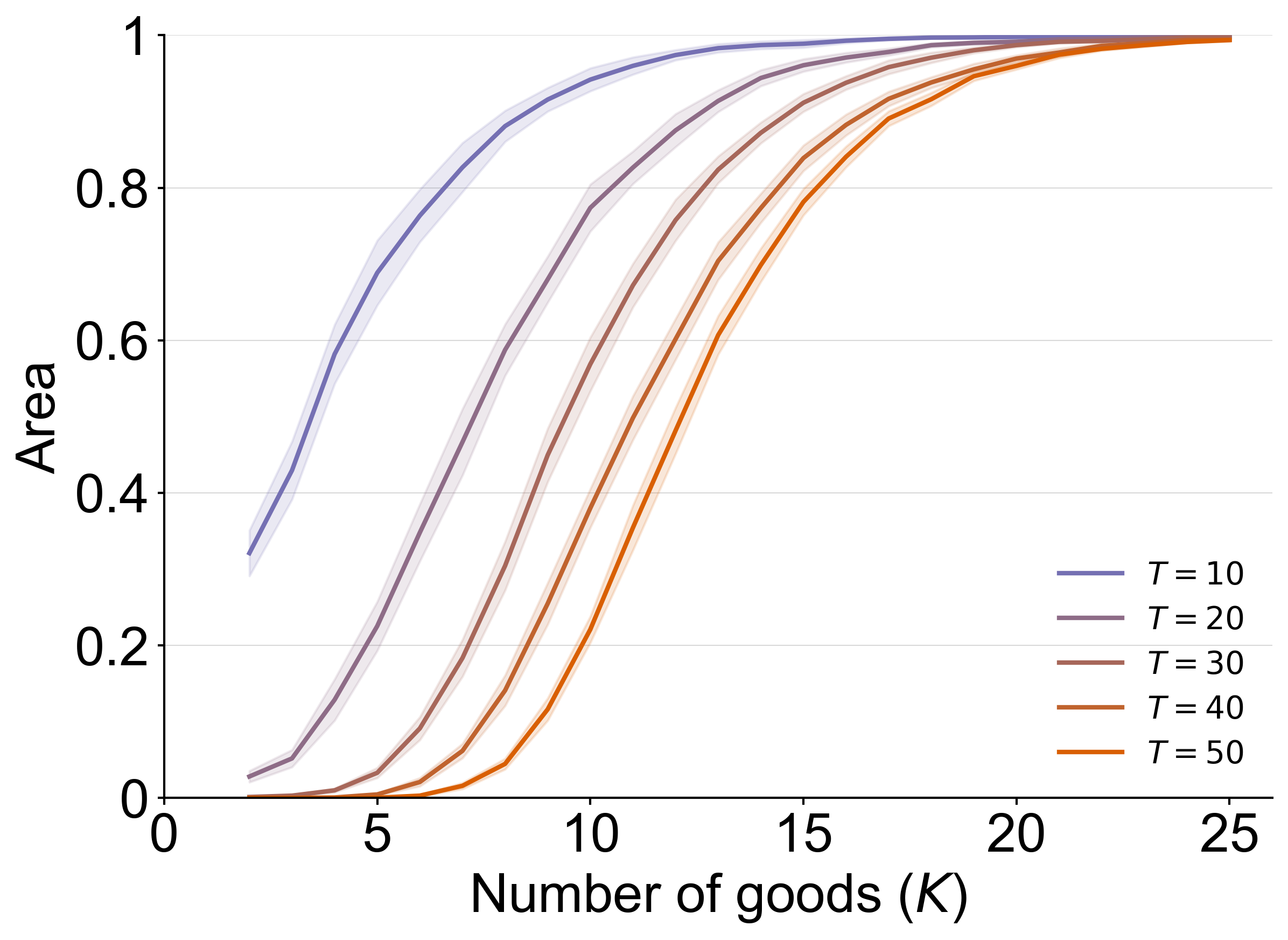}
        \caption{Area and $K$, by $T$.}\label{fig:AK_by_T}
    \end{subfigure}\hfill
    \begin{subfigure}[b]{0.5\textwidth}
    \centering
    \includegraphics[width=1.0\textwidth]{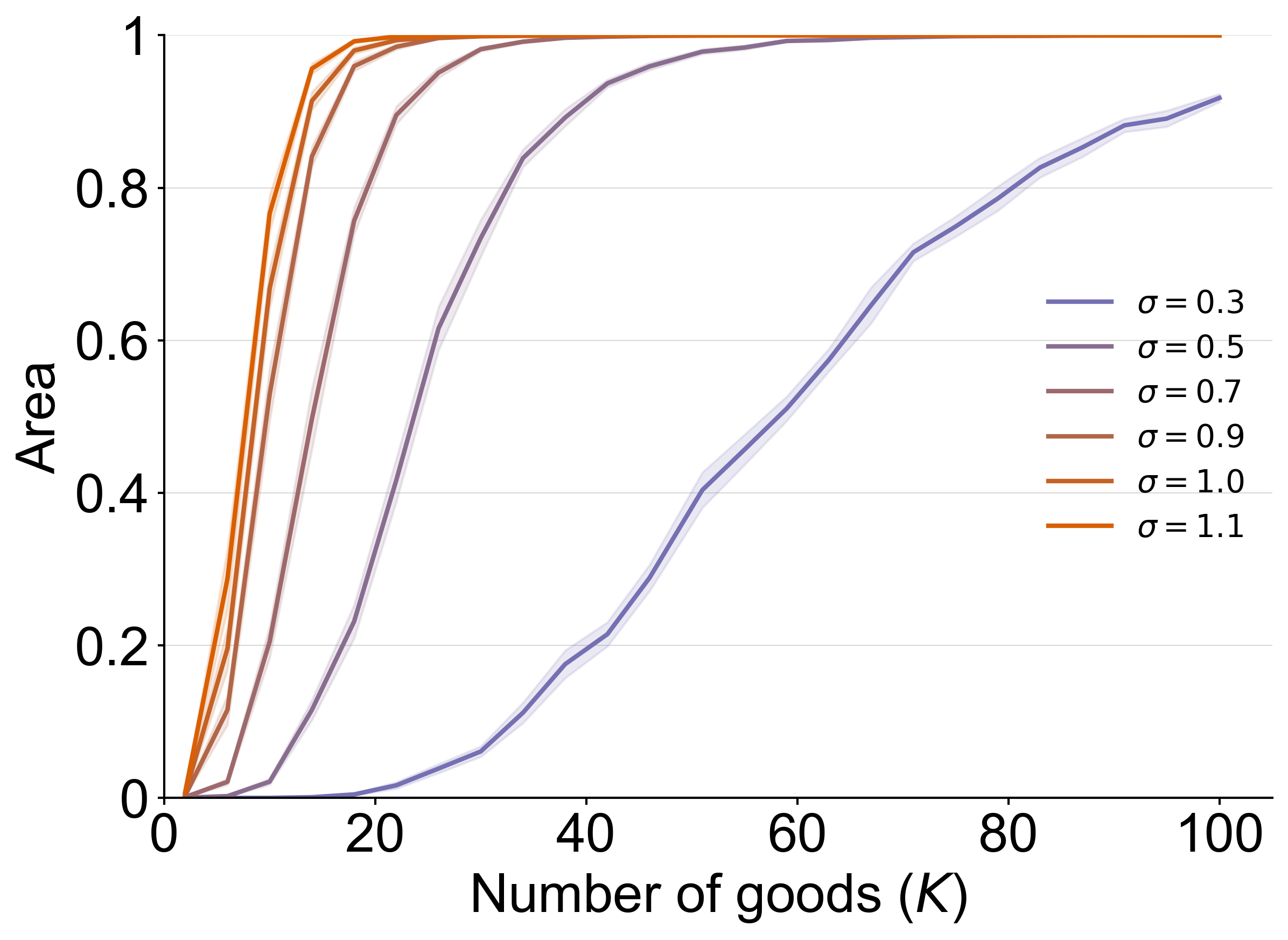}
    \caption{Area and $K$, by $\sigma$.}\label{fig:AK_by_sigma}
    \end{subfigure}
    \note[Note]{ \ref{fig:AK_by_T} shows the effect as the number of observations increases using the benchmark $\text{LogN}(0,1)$; \ref{fig:AK_by_sigma} shows the effect of varying the standard deviation of $\text{LogN}(0,\sigma)$. Inverse normalised prices are drawn as $1/r^k \sim \text{LogN}(0,\sigma)$. Panel (a): $T \in \{10, 20, 30, 40, 50\}$, $\sigma = 1$, $K$ up to $25$. Panel (b): $\sigma \in \{0.3, 0.5, 0.7, 0.9, 1.0, 1.1\}$, $T = 25$, $K$ up to $100$. Each $(K, T, \sigma)$ point uses 20 price replications and adaptive Monte Carlo with up to 10{,}000 budget-share draws per replication. Shaded bands are 95\% confidence intervals across replications.}
    \label{fig:AK}
\end{figure}
Panel~\ref{fig:AK_by_T} illustrates the effect of increasing the number of goods conditional on the number of observations. It shows that increasing the number of observations slows, but does not prevent, the loss of empirical content. It also shows that, though the effects of adding observations are attenuated as $T$ increases--the $A(K)$ curves are more closely spaced as $T$ increases. Under the benchmark $\text{LogN}(0,1)$ with $T = 10$, the Area exceeds $0.9$ by $K = 9$ and $0.95$ by $K = 11$. Doubling $T$ to $20$ shifts these thresholds only to $K = 13$ and $K = 15$ respectively; even at $T = 50$, $A_K$ reaches $0.95$ by $K = 20$.

Panel~\ref{fig:AK_by_sigma} illustrates the effect of increasing the number of goods conditional on the dispersion of inverse normalised prices. Lower dispersion in relative prices reduces the Area---having budget constraints tightly clustered so that (i) they cross frequently and (ii) the relative prices are similar but not identical---but the effect of dimensionality remains dominant. At $T = 25$, reducing $\sigma$ from $1.0$ to $0.5$ shifts the $K$ threshold for $A_K \ge 0.9$ from $K = 14$ to $K = 42$.

Our overall takeaway is that responses such as increasing the number of observations or relying on realistic levels of price dispersion do not prevent the rapid loss of empirical content.

\subsection{The Area under Dimension-reducing Restrictions}

A natural response to the loss of empirical content is to impose additional structure that reduces the effective dimensionality of the problem. A standard approach is to assume separable preferences, which partition goods into groups and represent choices through lower-dimensional subutilities.

\cite{varian1983non}  provides necessary and sufficient conditions for rationalisation under both weak and additive separability. Unfortunately the conditions for weak separability are not known to admit a polynomial-time test.  We therefore work with simple necessary conditions—namely that the data satisfy GARP both overall and within each group. The Areas presented for the weakly separable case should therefore be interpreted as upper bounds on the true Areas. For additive separability, by contrast, the conditions are linear, and we implement the corresponding necessary and sufficient conditions directly.

The effect of separability on the Area operates through three channels. First, each commodity group contains fewer than $K$ goods, so the Areas associated with the subutilities are smaller than that of the full $K$-good problem. Second, the macro utility maps from $\mathbb{R}^{\tfrac{K}{G}} \rightarrow \mathbb{R}$ rather than from $\mathbb{R}^K$ to $\mathbb{R}$ (where $G$ denotes the group size, so the number of groups is $\tfrac{K}{G}$). This further reduces the Area. Third, both the within-group and across-group conditions must be satisfied, so a violation in any component leads to rejection of separability. Taken together, these restrictions make separable models more restrictive than the non-separable case.

\begin{figure}[ht]
\centering
\caption{Area under separability restrictions}
\includegraphics[width=0.6\textwidth]{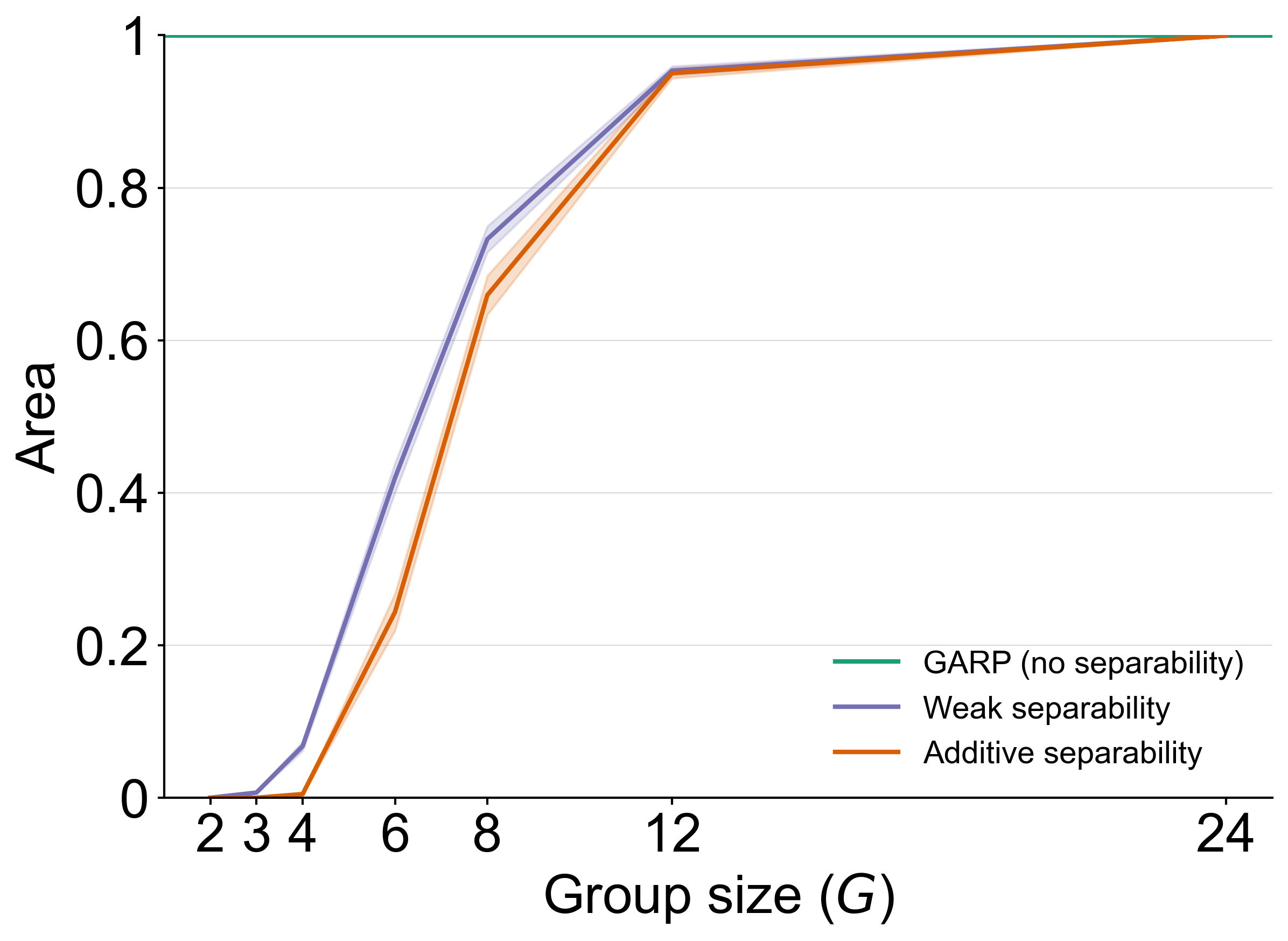}

\note[Note]{$K=24$ goods, $T=10$ observations, prices drawn as $1/r^k \sim \text{LogN}(0,1)$. The green line shows the Area for the non-separable model. The purple and orange lines show the $A(G)$ curve under weak and additive separability, respectively, as a function of group size $G$ and equally-sized groups. For each $G$, 100 random equal partitions are averaged over 10 independent price replications, with adaptive Monte Carlo up to 10{,}000 budget-share draws per partition. Shaded bands are 95\% confidence intervals.}
\label{f:separability}
\end{figure}

Figure~\ref{f:separability} compares the $A(G)$ curves under weak and additive separability assumptions for different group sizes $G$ (with $K/G$ equal-sized groups) and a fixed number of goods ($K=24$). As $G$ increases (i.e., as we have fewer, larger groups), the Area increases in a similar manner to the way it does with the number of goods. However, with $K=24$ and $T=10$, the Area for the unrestricted model is essentially one. The reduction in Area due to separability is therefore economically significant. Recalling that the Area for weak separability is an upper bound (as only the necessary conditions are computationally feasible), the difference between the $A(G)$  curves for weak and additive separability is clearly modest. At $K = 24$ and $T = 10$ the unrestricted Area is essentially one ($0.9996$). Separability into 3 groups of 8 goods ($G = 8$) reduces it to $0.73$ (weak) and $0.66$ (additive); reducing to 6 groups of 4 ($G = 4$) further drops the Area to $0.07$ (weak) and $0.005$ (additive).

Our overall takeaway is that separability can meaningfully increase empirical content, but the gains depend on the degree of aggregation and may be limited.

\subsection{The Design of Experiments}

In experimental settings the Area of rational choice is controlled through the design of the budget constraints. Experimental tests of rationality have mostly been in two-dimensional settings but recent experimental designs have increased dimensionality to make the tasks more cognitively demanding  (see \cite{Ahn_2014}, \cite{Halevy_2024} and \cite{Dembo_2025}). This may reduce empirical content.

In the literature there are two broad approaches to experimental design: a fixed design and an adaptive design. The first, and most standard, typified by \cite{Choi_2007} and \cite{Choi_2014}, pre-selects a fixed set of budget constraints to be presented to experimental subjects. The second uses an adaptive design in which budget constraints are chosen sequentially, conditional on the observed actions of the subject. Specifically, the subject is given an initial budget constraint and makes a choice. The next constraint is then selected so that the first bundle is exactly affordable, and so on.  This approach originates in \cite{Sippel_1997}  and is developed further in  \cite{BBC_2003}.
\begin{figure}[!htb]
    \centering
     \caption{Area and experimental design. }
    \begin{subfigure}[b]{0.5\textwidth}
        \centering
       \includegraphics[width=1.0\textwidth]{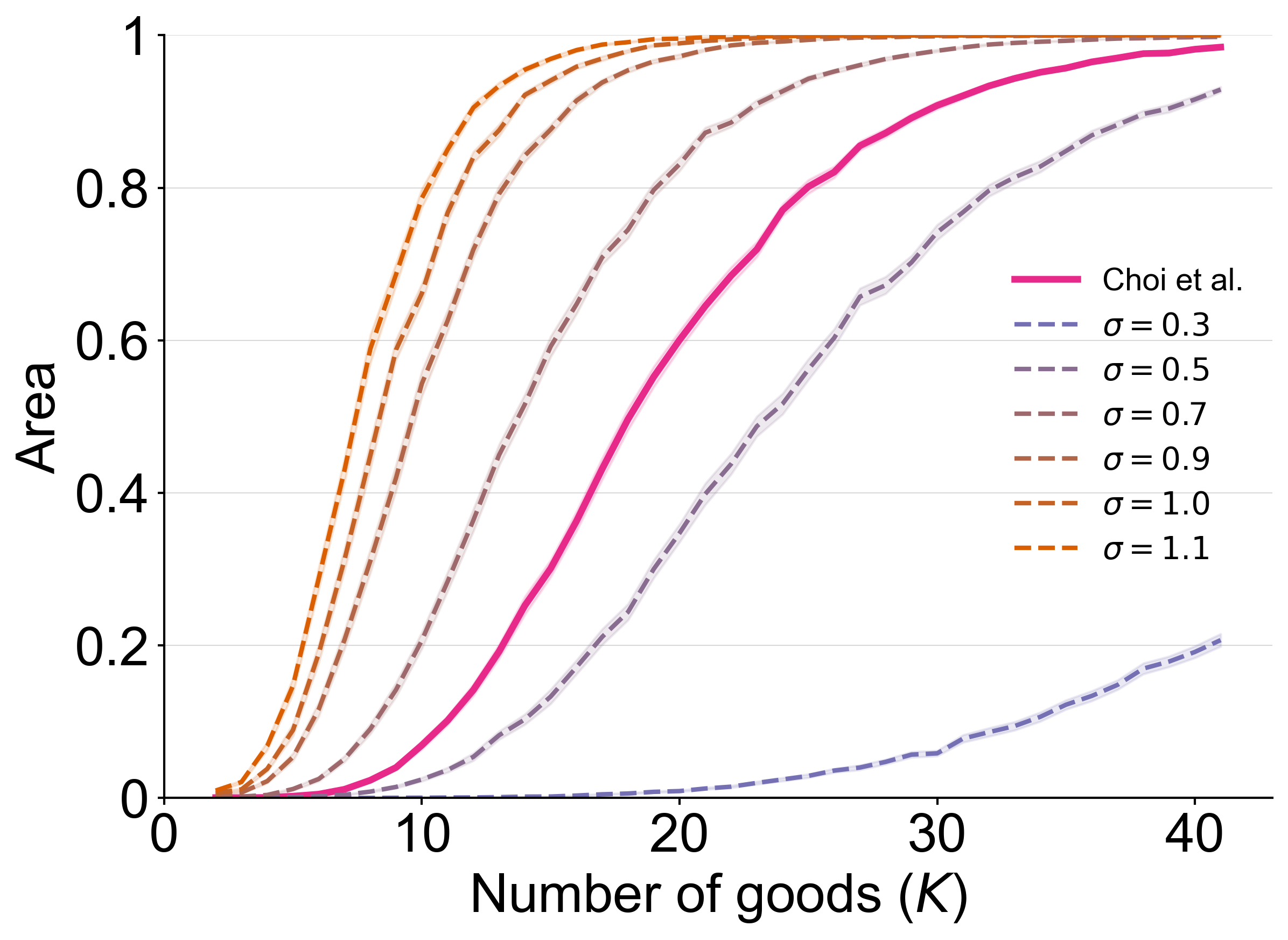}
        \caption{Fixed design.}\label{fig:Choi}
    \end{subfigure}\hfill
    \begin{subfigure}[b]{0.5\textwidth}
    \centering
    \includegraphics[width=1.0\textwidth]{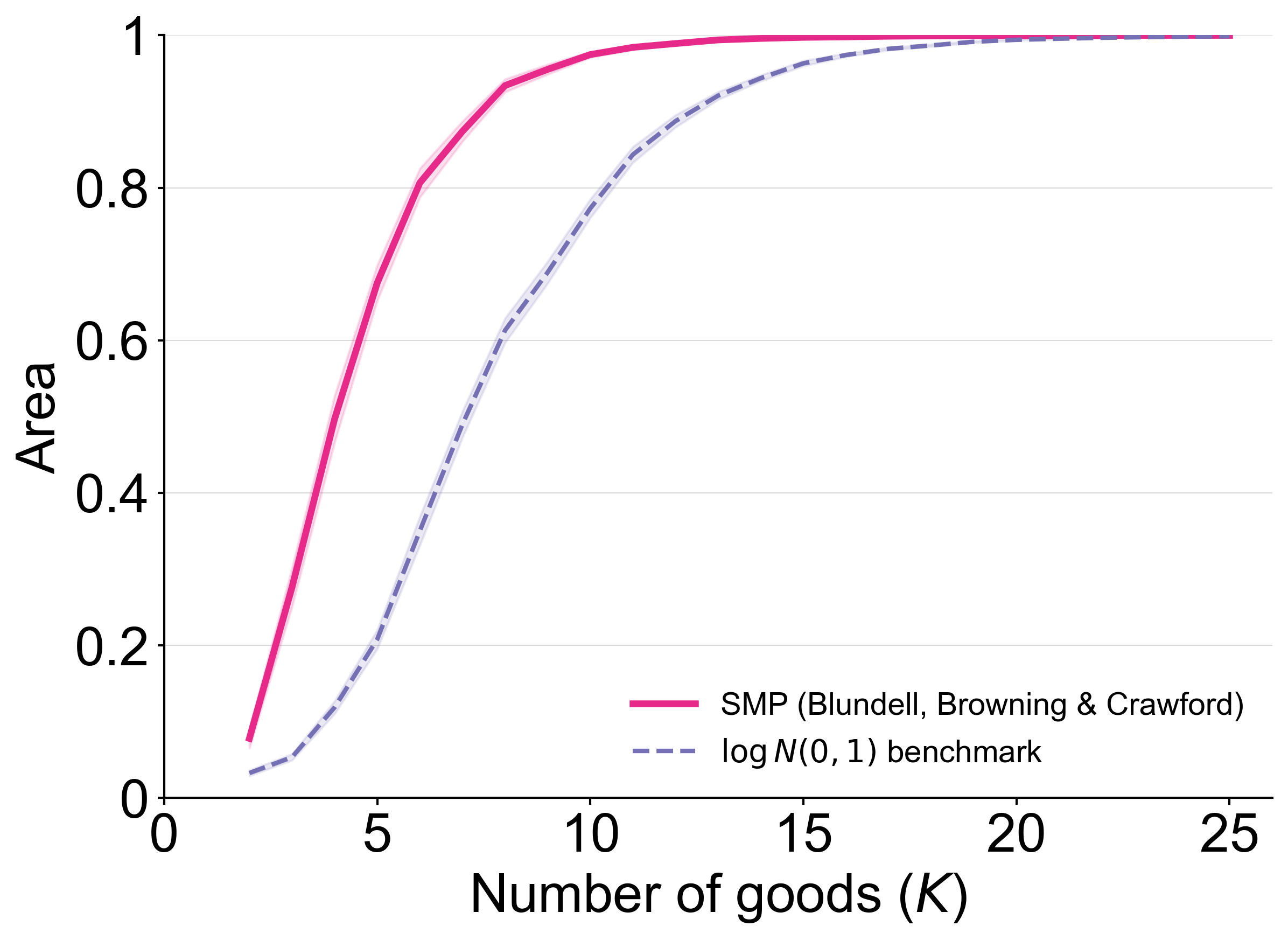}
    \caption{Choice-adapted design.}\label{fig:BBC}
    \end{subfigure}
    \note[Note]{\ref{fig:Choi} shows Area against $K$ for the design in \cite{Choi_2014} with $T=25$ and the Area under the benchmark varying the standard deviation of $\text{LogN}(0,\sigma)$. \ref{fig:BBC} shows the Area against the number of goods $(K)$ using the design in \cite{BBC_2003} and the benchmark $\text{LogN}(0,1)$. Inverse normalised prices for the LogN benchmarks are drawn as $1/r^k \sim \text{LogN}(0, \sigma)$. Panel (a): \cite{Choi_2014} design vs LogN$(0, \sigma)$ benchmarks for $\sigma \in \{0.3, 0.5, 0.7, 0.9, 1.0, 1.1\}$, $T = 25$, $K$ up to $40$. Panel (b): \cite{BBC_2003} SMP design vs LogN$(0, 1)$ benchmark, $T = 20$, $K$ up to $25$. See footnotes for design rules. Both panels use 100 price replications and adaptive Monte Carlo with up to 50{,}000 budget-share draws per replication. Shaded bands are 95\% confidence intervals.}
    \label{fig:Design}
\end{figure}

Figure \ref{fig:Design} shows simulations of the  \cite{Choi_2014}  design (Panel \ref{fig:Choi})\footnote{The design is as follows. Inverse normalised prices (intercepts) are chosen initially randomly in the interval $[a,b]$. Then budget constraints where both intercepts are in $[a, \tfrac{b}{2}]$ are discarded. } and the extension of the  \cite{Sippel_1997}
 design by \cite{BBC_2003} (Panel \ref{fig:BBC}).\footnote{The design is as follows. Initially an inverse normalised price is drawn from the benchmark $\text{LogN}(0,1)$. A bundle is selected on that constraint from a uniform distribution. The next vector of inverse normalised prices is drawn from the benchmark distribution and scaled such that $\bm{r}_2\cdot\bm{x}_1 = \bm{r}_1\cdot\bm{x}_1$. The second bundle is then selected from the new budget constraint and so on.} 

Looking first at the fixed design we see that the design performs well against the $\text{LogN}(0,1)$ benchmark. It sits between the $A(K)$ curves for $\text{LogN}(0,0.5)$ and $\text{LogN}(0,0.7)$. At $K = 20$ the Choi design's Area is $0.60$, between $\text{LogN}(0,0.5)$ (Area $0.35$) and $\text{LogN}(0,0.7)$ (Area $0.83$); the $\text{LogN}(0,1)$ benchmark already has Area $0.99$ at $K = 20$. However, we also see that it does less well against inverse normalised price distributions with much smaller standard deviations. Looking at the adaptive design we see that the $A(K)$ curve for the SMP path lies above the benchmark. At $K = 10$ the SMP path's Area is $0.97$ against $0.77$ for the $\text{LogN}(0,1)$ benchmark; by $K = 20$ SMP reaches $0.9996$ against the benchmark's $0.99$. It appears that, whilst the SMP path is powerful conditional on a given sequential preference ordering, it is associated with large Areas and hard-to-reject tests of rational choice compared to the simple benchmark. The reason is that if one of the normalised price vectors is associated with ``extreme'' relative price difference, this can cause the following budget constraint to move outwards or inwards to a degree that reduces the extent to which budget constraints cross. Such designs are therefore best suited to environments with limited price variation, or where constraints are carefully sequenced.

Design choices materially affect empirical content, but even carefully constructed designs do not eliminate the dimensionality problem; designs effective in low dimensions may perform poorly in high dimensions.

\section{Concluding Remarks}
 
Revealed preference theory imposes nonparametric testable restrictions on observed behaviour. This paper studies how informative these restrictions are in high-dimensional choice environments. Using Selten’s Area measure, we show that for any fixed number of observations, the empirical content of GARP declines exponentially in the number of goods. As a result in high-dimensional settings almost all possible choices which satisfy budget constraints will appear rational. The mechanism underlying this result is a high-dimensional concentration phenomenon which appears in both the graph-theoretic and linear programming formulations. 

The simulation results show that this effect is not merely asymptotic but quantitatively important in empirically relevant environments. Calibrating the distribution of normalised prices to household scanner data, we find that empirical content deteriorates rapidly even at moderate numbers of goods. Increasing the number of observations slows this decline but does not prevent it. Similarly, realistic levels of price dispersion reduce the Area only modestly. Taken together, these findings indicate that standard empirical environments leave limited scope for revealed preference tests to reject behaviour inconsistent with GARP.

We also consider potential responses. Imposing additional structure, such as separability, reduces the effective dimensionality of the problem and can partially restore empirical content. Experimental design choices can likewise affect the extent to which budget constraints intersect and therefore the restrictiveness of the test. However, these approaches mitigate rather than eliminate the underlying effect. The loss of empirical content is a fundamental consequence of high dimensionality.

More broadly, the results highlight a general trade-off between flexibility and testability in nonparametric models of behaviour. As the dimensionality of the choice environment increases, the restrictions imposed by revealed preference theory become less informative. Revealed preference tests may therefore have limited empirical content, making it difficult to distinguish rational behaviour from arbitrary choice.

 \bibliography{paper.bib}

\section{Proofs}\label{s:proofs}

\subsection{Proof of Lemma~\ref{l:cycle_mean_inequality}}\label{a:cme_proof}
\begin{proof}Given their definition, on any edge $\ell$, the ratios of normalised prices $\{\rho_\ell^k\}_{k\in [K], \ell\in [L]}$ are positive reals satisfying telescoping: $\prod_{\ell=1}^L \rho_\ell^k =1$ for all $k\in [K]$. Write $v_\ell^k = \log \rho_\ell^k$, so that the telescoping property becomes $\sum_{\ell=1}^L v_\ell^k = 0$ for each $k$. By Jensen's inequality,
\begin{equation}
    \bar{\rho}_\ell = \frac{1}{K}\sum_{k=1}^K e^{v_\ell^k} \ge \exp\left(\frac{1}{K}\sum_{k=1}^K v_\ell^k\right),
\end{equation}
with equality if and only if $v_\ell^k$ is constant in $k$. Summing $\log \bar{\rho}_\ell$ over $\ell$ gives $\sum_{\ell} \log \bar{\rho}_\ell \ge \frac{1}{K}\sum_k \sum_\ell v_\ell^k = 0$, and exponentiating yields $\prod_{\ell=1}^L \bar{\rho}_\ell \ge 1$. Equality in the final inequality holds if and only if equality holds in Jensen's inequality for every $\ell$, which is equivalent to $v_\ell^k$, and hence $\rho_\ell^k$, being constant in $k$ for each $\ell$. Therefore, if $\rho_\ell^k$ is non-constant in $k$ for some $\ell$, at least one Jensen inequality is strict, so $\prod_{\ell=1}^L \bar{\rho}_\ell>1$. In that case $\max_{\ell\in[L]}\bar{\rho}_\ell>1$, since otherwise all $\bar{\rho}_\ell\le1$ and the product would be at most one.
\end{proof}

\subsection{Proof of Theorem~\ref{t:garp_high_dim}}\label{a:main_proof}

We first establish two lemmas. Lemma~\ref{l:strict_bounds} shows that~(\ref{ass:A1}) and~(\ref{ass:A2}) together imply $a<1<b$. Lemma~\ref{l:edge_prob_bound} provides a concentration bound showing that, under uniformly distributed budget shares, the probability that a weighted average of price ratios falls below 1 decays exponentially in $K$. The proof of the theorem then proceeds as follows. GARP is violated if and only if $G_K$ contains a directed cycle. For any cycle,~(\ref{ass:A2}) identifies an edge $\ell^*$ with $\bar{\rho}_{\ell^*}\ge 1+\varepsilon$. Lemma~\ref{l:edge_prob_bound} shows that such an edge forms with probability at most $\exp(-c_1K)$. A union bound over the finitely many cycles on $[T]$ completes the argument.

\begin{lemma}[manual-num=2, label={l:strict_bounds}][Strict bounds]
Under~(\ref{ass:A1}) and~(\ref{ass:A2}) of Theorem~\ref{t:garp_high_dim}, $a<1<b$.
\end{lemma}

\begin{proof}
The identity $\rho_{ij}^k\cdot \rho_{ji}^k = 1$ and~(\ref{ass:A1}) give $a\le \rho_{ij}^k \le b$ and $a\le 1/\rho_{ij}^k\le b$ for all $i\neq j$ and $k$. If $a\ge 1$, then $\rho_{ij}^k\ge 1$ and $\rho_{ji}^k = 1/\rho_{ij}^k\le 1$, but also $\rho_{ji}^k\ge a\ge 1$, forcing $\rho_{ij}^k=1$ for all $i,j,k$ and contradicting~(\ref{ass:A2}). Hence $a<1$; a symmetric argument gives $b>1$.
\end{proof}

\begin{lemma}[manual-num=3, label={l:edge_prob_bound}][Edge probability bound]
Let $\pmb{\rho} = (\rho^1,\dots,\rho^K) \in \mathbb{R}^K_{++}$ satisfy $a \le \rho^k \le b$ for all $k$, and suppose
\[
\bar{\rho} := \frac{1}{K}\sum_{k=1}^K \rho^k > 1.
\]
If $\bm{w} \sim \mathrm{Uniform}(\Delta_{K-1})$, then
\[
\mathbb{P}(\pmb{\rho}\cdot \bm{w} \le 1)
\le
\exp\!\left(
-\frac{K(\bar{\rho}-1)^2}{4(b-a)^2}
\right).
\]
\end{lemma}

\begin{proof}
Since $\bm{w} \sim \mathrm{Uniform}(\Delta_{K-1}) = \mathrm{Dirichlet}(1,\dots,1)$, we may write
\[
w^k = \frac{Y_k}{S},
\qquad
Y_k \stackrel{\mathrm{iid}}{\sim} \mathrm{Exp}(1),
\qquad
S = \sum_{k=1}^K Y_k.
\]
Then $\pmb{\rho}\cdot \bm{w} \le 1$ if and only if
\[
Z := \sum_{k=1}^K (\rho^k-1)Y_k \le 0.
\]
Set $a_k = \rho^k - 1$. Since $\mathbb{E}[Y_k]=1$,
\[
\mathbb{E}[Z] = \sum_{k=1}^K a_k = K(\bar{\rho}-1) > 0.
\]

For any $t>0$ with $t a_k > -1$ for all $k$, Markov's inequality gives
\[
\mathbb{P}(Z \le 0)
=
\mathbb{P}(e^{-tZ}\ge 1)
\le
\mathbb{E}[e^{-tZ}]
=
\prod_{k=1}^K \mathbb{E}[e^{-t a_k Y_k}]
=
\prod_{k=1}^K (1+t a_k)^{-1},
\]
using the Laplace transform of $\mathrm{Exp}(1)$.

Let $M = \max_k |a_k| \le b-a$. For $|x|\le 1/2$, we use
\[
\log(1+x) \ge x - x^2.
\]
Hence, whenever $|t a_k|\le 1/2$,
\[
\log \mathbb{P}(Z\le 0)
\le
-\sum_{k=1}^K \log(1+t a_k)
\le
-t\sum_{k=1}^K a_k + t^2 \sum_{k=1}^K a_k^2
\le
-tK(\bar{\rho}-1) + t^2 K M^2.
\]
Optimising at $t^* = (\bar{\rho}-1)/(2M^2)$ yields
\[
\log \mathbb{P}(Z\le 0)
\le
-\frac{K(\bar{\rho}-1)^2}{4M^2}.
\]
Since $M \le b-a$, the stated bound follows.
\end{proof}

\begin{proof}[Proof of Theorem~\ref{t:garp_high_dim}]
The number of directed cycles on distinct vertices in $[T]$ is finite and given by $C_T = \sum_{L=2}^T \binom{T}{L}(L-1)!<\infty.$ Fix a directed cycle $C = (i_1,\ldots,i_L,i_1)$ on distinct vertices. A GARP violation occurs if and only if the revealed preference graph contains a directed cycle. For any such cycle, assumption~(\ref{ass:A2}) implies that there exists an edge $\ell^*$ such that $\bar{\rho}_{\ell^*} \ge 1+\varepsilon$. The event $\{i_\ell\to i_{\ell+1}\} = \{\pmb{\rho}_{i_\ell i_{\ell+1}}\cdot \bm{w}_{i_{\ell+1}}\le 1\}$ depends only on $\bm{w}_{i_{\ell+1}}$. Since the $\bm{w}_{i}$ are independent, the $L$ edge events are mutually independent:
\begin{equation}
    \mathbb{P}(C\text{ appears in }G_K) = \prod_{\ell=1}^L \mathbb{P}(\pmb{\rho}_{i_\ell i_{\ell+1}}\cdot \bm{w}_{i_{\ell+1}}\le 1).
\end{equation}
Applying Lemma~\ref{l:edge_prob_bound}:
\begin{equation}
    \mathbb{P}(C\text{ appears in }G_K)\le \exp\left(-\frac{K\varepsilon^2}{4(b-a)^2}\right) =: \exp(-c_1K).
\end{equation}
A union bound over all $C_T$ cycles gives $\mathbb{P}(\text{GARP violation})\le C_T\exp(-c_1K)$, and hence $A_K \ge 1 - C_T\exp(-c_1K)$.
\end{proof}

\subsection{Proof of Theorem~\ref{t:lp_vanishing}}\label{a:lp_proof}
The proof of Theorem~\ref{t:lp_vanishing} combines two ingredients: a concentration bound (Lemma~\ref{l:concentration}) showing that each LP coefficient $e_{ij}$ is exponentially unlikely to fall below its mean $\bar{\rho}_{ij}$, and a strict feasibility result (Lemma~\ref{l:strict_feasibility}) showing that the LP evaluated at the limiting coefficients admits a solution with uniform slack $\eta_0>0$. When the random coefficients stay within this slack---an event whose complement has probability at most $T(T-1)\exp(-c_2K)$---the limiting solution remains feasible for the random LP.

\begin{lemma}[manual-num=4, label={l:concentration}][Concentration of LP coefficients]
Under~(\ref{ass:A1}), for each $i\neq j$ and any $\delta > 0$,
\begin{equation}\label{e:conc_bound}
    \mathbb{P}(e_{ij} < \bar{\rho}_{ij} - \delta) \le \exp\!\left(-\frac{K\delta^2}{4(b-a+\delta)^2}\right).
\end{equation}
\end{lemma}

\begin{proof}
Since $\text{Uniform}(\Delta_{K-1}) = \text{Dirichlet}(1,\ldots,1)$, we may write $w_j^k = Y_k / S$ with $Y_k \sim \text{Exp}(1)$ independent and $S = \sum_{k=1}^K Y_k$ \citep[Chapter~11, Theorem~4.1]{devroyeMultivariateDistributions1986}. Since $S > 0$ a.s., $e_{ij} < \bar{\rho}_{ij} - \delta$ if and only if
\begin{equation}
    Z := \sum_{k=1}^K \bigl(\rho_{ij}^k - \bar{\rho}_{ij} + \delta\bigr)\, Y_k \le 0.
\end{equation}
Setting $\alpha_k = \rho_{ij}^k - \bar{\rho}_{ij} + \delta$, we have $\sum_{k} \alpha_k = K\delta > 0$ and $|\alpha_k| \le (b-a)+\delta =: M_\delta$. For $\theta > 0$ with $\theta M_\delta \le 1/2$, Markov's inequality and $\mathbb{E}[e^{-\theta \alpha_k Y_k}] = (1+\theta \alpha_k)^{-1}$ give
\begin{equation}
    \mathbb{P}(Z\le 0) \le \prod_{k=1}^K (1+\theta \alpha_k)^{-1}.
\end{equation}
The bound $\log(1+x) \ge x - x^2$ for $|x|\le 1/2$ yields
\begin{equation}
    \log \mathbb{P}(Z\le 0) \le -\theta K\delta + \theta^2 K M_\delta^2.
\end{equation}
Setting $\theta^* = \delta/(2M_\delta^2)$ (which satisfies $\theta^* M_\delta \le 1/2$ since $\delta \le M_\delta$) gives $\log\mathbb{P}(Z\le 0) \le -K\delta^2/(4M_\delta^2)$.
\end{proof}

\begin{lemma}[manual-num=5, label={l:strict_feasibility}][Strict feasibility of the limiting LP]
Under~(\ref{ass:A1}) and~(\ref{ass:A2p}), define $\eta_0 = \eta/(2T)$. There exist $U_1^*,\ldots,U_T^*\in \mathbb{R}$ such that
\begin{equation}\label{e:tightened}
    U_j^* - U_i^* \le (\bar{\rho}_{ij} - 1) - \eta_0, \quad \text{for all } i,j\in [T], \; i\neq j.
\end{equation}
\end{lemma}

\begin{proof}
Set $d_{ij} = \bar{\rho}_{ij} - 1 - \eta_0$. By~(\ref{ass:A2p}), every sequence of distinct indices $(i_\ell)_{\ell\in [L]}$ with $L\ge 2$ and $i_{L+1}:=i_1$ satisfies
\begin{equation}\label{e:cycle_sum}
    \sum_{\ell=1}^L d_{i_\ell i_{\ell+1}} = \sum_{\ell=1}^L (\bar{\rho}_{i_\ell i_{\ell+1}} - 1) - L\eta_0 \ge \eta - T\eta_0 = \frac{\eta}{2} > 0,
\end{equation}
since $L\le T$ and $\eta_0 = \eta/(2T)$. Fix $U_1^* = 0$ and for each $m\neq 1$ define
\begin{equation}
    U_m^* = \min\!\left\{\sum_{p=1}^{n-1} d_{j_p j_{p+1}} : j_1 = 1,\; j_n = m,\; j_1,\ldots,j_n\in [T] \text{ distinct}\right\},
\end{equation}
a minimum over a finite set. To verify $U_j^* \le U_i^* + d_{ij}$: let $j_1=1,\ldots,j_n=i$ be a sequence achieving $U_i^*$. If $j\notin \{j_1,\ldots,j_n\}$, appending $j_{n+1}=j$ gives a distinct sequence from~$1$ to~$j$ with sum $U_i^*+d_{ij}$, so $U_j^* \le U_i^* + d_{ij}$. If $j = j_q$ for some $q < n$, the subsequence $j_1,\ldots,j_q$ is a sequence of distinct indices from~$1$ to~$j$, so $U_j^* \le \sum_{p=1}^{q-1} d_{j_p j_{p+1}} = U_i^* - \sum_{p=q}^{n-1} d_{j_p j_{p+1}}$. Since $\sum_{p=q}^{n-1} d_{j_p j_{p+1}} + d_{ij} = \sum_{p=q}^{n} d_{j_p j_{p+1}}$ with $j_{n+1}=j=j_q$ is a cyclic sum, it is non-negative by~\eqref{e:cycle_sum}, giving $U_j^* \le U_i^* + d_{ij}$. 
\end{proof}

\begin{proof}[Proof of Theorem~\ref{t:lp_vanishing}]
Let $(U_1^*,\ldots,U_T^*)$ be the solution from Lemma~\ref{l:strict_feasibility}, and set $\lambda_i = 1$. We show that $(U^*,\lambda)$ is feasible for the random Afriat system whenever $e_{ij} \ge \bar{\rho}_{ij} - \eta_0$ for all $i\neq j$.

For each constraint $(i,j)$ with $i\neq j$: by~\eqref{e:tightened}, $U_j^* - U_i^* \le (\bar{\rho}_{ij} - 1) - \eta_0$. If $e_{ij} \ge \bar{\rho}_{ij} - \eta_0$, then
\begin{equation}
    e_{ij} - 1 \ge (\bar{\rho}_{ij} - 1) - \eta_0 \ge U_j^* - U_i^*,
\end{equation}
so the constraint $U_j^* - U_i^* \le \lambda_i(e_{ij}-1) = e_{ij}-1$ holds. By a union bound over all $T(T-1)$ pairs and Lemma~\ref{l:concentration} with $\delta = \eta_0$:
\begin{equation}
    \mathbb{P}(\text{LP infeasible}) \le \mathbb{P}\!\left(\exists\, i\neq j : e_{ij} < \bar{\rho}_{ij} - \eta_0\right) \le T(T-1)\exp\!\left(-\frac{K\eta_0^2}{4(b-a+\eta_0)^2}\right).
\end{equation}
Setting $c_2 = \eta_0^2/(4(b-a+\eta_0)^2) > 0$, where $\eta_0 = \eta/(2T)$, completes the proof.
\end{proof}
\newpage
\section{The distribution of normalised prices}
\renewcommand{\thefigure}{B\arabic{figure}}
\setcounter{figure}{0}

We calibrate the distribution of prices using the ACNielsen Homescan Panel (\cite{Aguiar_Hurst_2007}), which provides detailed household-level price and expenditure data. Figure \ref{f:price_distr} shows the results of fitting a normal distribution to the log inverse normalised prices, $\log(\tfrac{1}{r^k})$. The top right panel shows the fitted distribution against the histogram of the data, the top left shows the q-q plot of the quantiles of the empirical and fitted distributions; bottom left shows the empirical and fitted CDF and bottom right shows the p-p plot of theoretical cumulative probabilities from the fitted distribution and the data. We find that  $\tfrac{1}{r^k} \sim \text{LogN}(5.69,1.19)$ is a good description of the real-world data in \cite{Aguiar_Hurst_2007}.
\begin{figure}[h]
    \caption{The distribution of inverse normalised prices}
    \label{f:price_distr}
    \centering
    \includegraphics[width=0.6\textwidth]{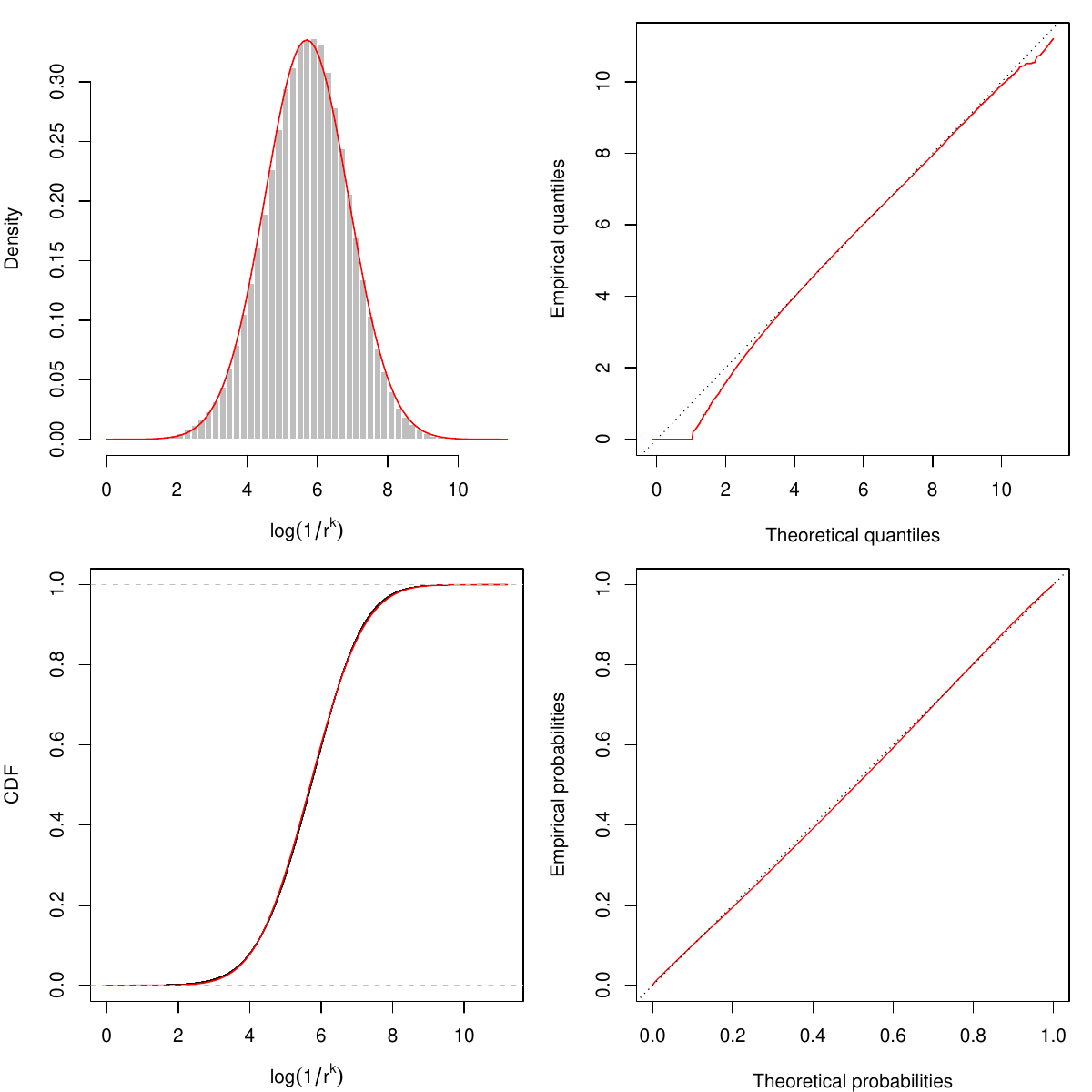}
\end{figure}

\end{document}